\documentclass[conference]{IEEEtran}
\IEEEoverridecommandlockouts
\usepackage{cite}
\usepackage{amsmath,amssymb,amsfonts}
\usepackage{algorithmic}
\usepackage{graphicx}
\usepackage{textcomp}
\usepackage{xcolor}
\usepackage{bm}
\usepackage{booktabs,amsthm}
\usepackage{subcaption} 
\usepackage[numbers,sort&compress]{natbib}
\def\BibTeX{{\rm B\kern-.05em{\sc i\kern-.025em b}\kern-.08em
    T\kern-.1667em\lower.7ex\hbox{E}\kern-.125emX}}

\usepackage{comment}

\newtheorem{lemma}{Lemma}
\newtheorem{corollary}{Corollary}
\newtheorem{theorem}{Theorem}
\newtheorem{definition}{Definition}
\newtheorem{example}{Example}
\newtheorem{remark}{Remark}

\begin{document}

\title{New Results Towards the Characterization of Service Rate Region of Reed-Muller Codes
\thanks{Research supported by National Key Research and Development Program of China under Grant Nos. 2022YFA1004900 and 2021YFA1001000, National Natural Science Foundation of China under Grant No. 12231014, and the Taishan Scholars Program. Corresponding author: Y. Zhang.}
}

\author{\IEEEauthorblockN{{Chenyang Yu},
{Luopeng Sun} and
{Yiwei Zhang}}
\IEEEauthorblockA{State Key Laboratory of Cryptography and Digital Economy Security,\\
School of Cyber Science and Technology, Shandong University, Qingdao, Shandong, 266237, China\\
cyyu@mail.sdu.edu.cn; lps\_sdu@mail.sdu.edu.cn; ywzhang@sdu.edu.cn}
}

\maketitle

\begin{abstract}
The Service Rate Region (SRR) serves as a critical metric for evaluating the concurrent service capacity of distributed storage systems. While several works have characterized the SRR for MDS codes and first order Reed-Muller codes, for high-order Reed-Muller codes the 
problem becomes way more complicated and only partial results were given by Ly, Soljanin, and Lalitha [IEEE ISIT 2025].
In this paper, we refine the SRR analysis by explicitly characterizing the intersection patterns of recovery sets for high-order Reed-Muller codes, deriving the exact region for the case \(m=r+1\) and providing new types of strictly tighter constraints to bridge the gap between existing approximations and the exact SRR polytope.
\end{abstract}

\section{Introduction}

The scalability of modern distributed storage systems (DSS) fundamentally depends on serving concurrent access requests. While replication often falters under skewed demand~\cite{ref1_aktacs2021service}, coded redundancy offers superior throughput optimization. To rigorously evaluate these schemes, the Service Rate Region (SRR) serves as a critical metric. Defined as the set of all serviceable request rate vectors~\cite{ref1_aktacs2021service}, the SRR forms a convex polytope in $\mathbb{R}^{k}$ that precisely bounds the service capacity of a system storing $k$ independent message symbols~\cite{ref1_aktacs2021service}. Prior work has extensively characterized the SRR for specific classes of linear codes, including MDS codes~\cite{ref1_aktacs2021service,ref2_ly2026service}, Hamming codes~\cite{ref3_choudhary2025service}, and first-order Reed-Muller codes~\cite{ref4_kazemi2020geometric}. 
In addition, upper bounds on the maximum achievable service rate for a single data object under general linear coding schemes are derived in~\cite{ref5_ly2025maximal}.
However, extending these results to more complex structures presents non-trivial computational challenges, as approaches based on linear programming relaxations~\cite{ref6_alon2012large} become computationally demanding when recovery graph structures are not explicitly enumerated. 

Recently, the classic Reed-Muller (RM) codes~\cite{ref7_muller1954application,ref8_reed1954transactions} have attracted renewed attention in various domains~\cite{ref9_Kudekar2017reedmuller,ref10_reeves2023reed,ref11_abbe2023proof,ref12_calderbank2010construction,ref13_beimel2005general}, also including the analysis of its service rates. Most notably, Ly, Soljanin, and Lalitha~\cite{ref14_isit_ly2025service} provided a comprehensive analysis of the recovery-set structure of RM codes by leveraging their connections to finite Euclidean geometry. They derived explicit bounds on the maximum achievable demand for individual message symbols and established that the SRR can be approximated by two nested simplices -- a maximal achievable inner simplex and a minimal enclosing outer simplex -- which differ by at most a factor of two. Furthermore, they demonstrated that a complete and exact characterization of the SRR is computationally equivalent to the open coordinate-constrained weight enumerator problem for the dual code. While this work establishes a solid baseline using aggregate bounds and approximations for high-order RM codes, deriving the exact region invites a more granular examination of the interactions between message symbols of different orders and the precise boundaries of the polytope.

Motivated by this observation, we further investigate the SRR problem for high-order RM codes, giving a more precise characterization of the SRR for RM codes with arbitrary parameters, bridging the gap between the existing approximate simplices and the exact complete region. Building on the recovery-set properties of RM codes, we first obtain a complete characterization of the SRR for $\mathrm{RM}(r,r+1)$. 
Furthermore, for $\mathrm{RM}(r,m)$ codes with arbitrary parameters, 
we first derive upper bounds on the maximum achievable demand for any two message symbols in Theorems~\ref{thm:theorem2} and~\ref{thm:theorem3}, and then prove in Theorem~\ref{thm:theorem4} that these bounds are tight,
thereby characterizing all two-dimensional projections of the SRR and yielding a maximum achievable inscribed polytope.



\section{Preliminaries}

Let $[a,b]=\{a,a+1,\dots,b\}$ for positive integers $a$ and $b$, and $[1,n]$ is abbreviated as $[n]$. Let $q$ be a prime power and $\mathbb{F}_q$ be the finite field of order $q$. A $q$-ary linear $[n,k,d]_q$-code $\mathcal{C}$ is a $k$-dimensional subspace of the $n$-dimensional vector space $\mathbb{F}_q^{n}$ with minimum Hamming distance $d$. The notations $\boldsymbol{0}_k$ and $\boldsymbol{1}_k$ denote the all-zero and all-one column vectors of length $k$, respectively. The unit column vector with entry 1 in the $i$-th coordinate and 0's elsewhere is denoted by $\boldsymbol{e}_i$.

\subsection{Service Rate Region of Codes}

Consider a storage system consisting of $n$ servers, in which $k$ message symbols $s_1, \ldots, s_k$ are stored via a linear $[n, k,d]_q$-code with generator matrix $\boldsymbol{G} \in \mathbb{F}_q^{k \times n}$. Let $\boldsymbol{g}_j$ denote the $j$-th column of $\boldsymbol{G}$, for $1 \le j \le n$. A set $R \subseteq [n]$ is called a \emph{recovery set} for a message symbol $s_i$, if there exist coefficients ${\alpha_j}$'s from $\mathbb{F}_q$ such that
\(
\sum_{j \in R} \alpha_j \boldsymbol{g}_j = \boldsymbol{e}_i,
\)
i.e., the unit vector $\boldsymbol{e}_i$ is a linear combination of the column vectors of $\boldsymbol{G}$ indexed by $R$. A recovery set $R$ of $s_i$ is called \emph{minimal} if no proper subset of $R$ is still a recovery set of $s_i$. For a minimal recovery set $R$ the column vectors $\{\boldsymbol{g}_j:j\in R\}$ must be linearly independent and all the coefficients in the linear combination $\sum_{j \in R} \alpha_j \boldsymbol{g}_j= \boldsymbol{e}_i$ are non-zero. In the sequel whenever we say recovery sets we mean minimal recovery sets. Let $\mathcal{R}_i = \{ R_{i,1}, \ldots, R_{i,t_i} \}$ be the collection of all the $t_i$ recovery sets for $s_i$. 

For a storage system, multiple users would retrieve some certain message symbols and the servers should process these requests. For $j\in[n]$, let $\mu_j \in \mathbb{R}_{\ge 0}$ be the average rate at which the $j$-th server processes received requests. We refer to $\mu_j$ as the \emph{server capacity}. In most existing results and also the current paper, we assume that all servers have identical capacity which is normalized as $\mu_j=1$, for $j\in[n]$. 

For  $i \in [k]$, let $\lambda_i\in \mathbb{R}_{\ge 0}$ be the users' request rate to download the message symbol $s_i$, and let
\(
\boldsymbol{\lambda} = (\lambda_1, \ldots, \lambda_k) \in \mathbb{R}_{\ge 0}^k
\)
be the corresponding vector of request rates.
A scheduling strategy assigns a fraction of the requests for a message symbol to each of its recovery sets. 
Let $\lambda_{i,j}$ denote the portion of requests for $s_i$
assigned to the recovery set $R_{i,j}$, where $j \in [t_i]$.

The \emph{service rate region} (SRR), denoted by $\mathcal{S}(\boldsymbol{G}) \subseteq \mathbb{R}_{\ge 0}^k$,
is defined as the set of all request vectors $\boldsymbol{\lambda}$ that can be served by
a coded storage system with generator matrix $\boldsymbol{G}$ where each server has uniform capacity $\mu_j=1$. Therefore, $\mathcal{S}(\boldsymbol{G})$ is the set of all
non-negative vectors $\boldsymbol{\lambda}$ for which there exists an allocation
$\lambda_{i,j} \in \mathbb{R}_{\ge 0}$, $i \in [k]$ and $j \in [t_i]$,
satisfying the following constraints:
\[
\sum_{j=1}^{t_i} \lambda_{i,j} = \lambda_i, \quad \forall i \in [k], \tag{1a}\label{equa:1a}
\]
\vspace{-2mm}
\[
\sum_{i=1}^{k} \sum_{\substack{j=1,~\ell \in R_{i,j}}}^{t_i} \lambda_{i,j} \le 1, 
\quad \forall \ell \in [n], \tag{1b}\label{equa:1b}
\]
Here the first constraint means that the request rate $\lambda_i$ for the message symbol $s_i$ must be satisfied from all its recovery sets. The second constraint ensures that, for each server, the total amount of requests it receives does not exceed its capacity.

\subsection{Reed-Muller Codes and Their Properties}

Let $v_1, \ldots, v_m$ be $m$ binary variables forming an $m$-tuple
$\boldsymbol{v} = (v_1, \ldots, v_m)$.
For a Boolean function
\(
f(\boldsymbol{v}) = f(v_1, \ldots, v_m),
\)
let $\boldsymbol{f}$ be the vector consisting of all evaluations of $f$
over its $2^m$ possible arguments $\boldsymbol{v}$.

\begin{definition}\label{def:RMcodes}(\cite[Ch. 13]{ref15_macwilliams1977theory})
The $r$-th order binary Reed--Muller code $\mathrm{RM}(r,m)$ of length
$n = 2^m$, for $0 \le r \le m$, consists of all vectors $\boldsymbol{f}$ such that
$f(\boldsymbol{v})$ is a Boolean function that can be expressed as a polynomial
of degree at most $r$.
\end{definition}

It is well-known that $\mathrm{RM}(r,m)$ is a linear code of length $n = 2^m$, dimension
\(
k = \sum_{i=0}^{r} \binom{m}{i} \triangleq \binom{m}{\le r},
\)
and minimum distance
\(
d = 2^{m-r}
\)~\cite{ref15_macwilliams1977theory}.
Therefore, $\mathrm{RM}(r,m)$ is characterized by the parameters
\(
[n,k,d] = \bigl[2^m, \binom{m}{\le r}, 2^{m-r}\bigr]
\).
When $m \ge r+1$, the dual code of $\mathrm{RM}(r,m)$ is
$\mathrm{RM}(m-r-1,m)$~\cite{ref15_macwilliams1977theory}.
In the sequel, we assume $m \ge r+1$ to ensure the existence of $\mathrm{RM}(m-r-1,m)$. In the generator matrix $\boldsymbol{G_{\mathrm{RM}(r,m)}}$ of size $k\times n$, the rows correspond to Boolean functions. They are ordered by increasing degree, and within each degree, the monomials are listed in reverse lexicographic order of their variable indices. The columns simply follow a lexicographic order of all binary $m$-tuples.

\begin{example}
The generator matrix of $\mathrm{RM}(2,4)$ is given by
\[
\renewcommand{\arraystretch}{0.6}
\setlength{\arraycolsep}{1.6pt}
\scalebox{0.9}{
$
\boldsymbol{G}_{\mathrm{RM}(2,4)} =
\left[
\begin{array}{c|cccccccccccccccc}
\bm{1}
& 1 & 1 & 1 & 1 & 1 & 1 & 1 & 1 & 1 & 1 & 1 & 1 & 1 & 1 & 1 & 1 \\
\bm{v_4}
& 0 & 0 & 0 & 0 & 0 & 0 & 0 & 0 & 1 & 1 & 1 & 1 & 1 & 1 & 1 & 1 \\
\bm{v_3}
& 0 & 0 & 0 & 0 & 1 & 1 & 1 & 1 & 0 & 0 & 0 & 0 & 1 & 1 & 1 & 1 \\
\bm{v_2}
& 0 & 0 & 1 & 1 & 0 & 0 & 1 & 1 & 0 & 0 & 1 & 1 & 0 & 0 & 1 & 1 \\
\bm{v_1}
& 0 & 1 & 0 & 1 & 0 & 1 & 0 & 1 & 0 & 1 & 0 & 1 & 0 & 1 & 0 & 1 \\
\bm{v_3v_4}
& 0 & 0 & 0 & 0 & 0 & 0 & 0 & 0 & 0 & 0 & 0 & 0 & 1 & 1 & 1 & 1 \\
\bm{v_2v_4}
& 0 & 0 & 0 & 0 & 0 & 0 & 0 & 0 & 0 & 0 & 1 & 1 & 0 & 0 & 1 & 1 \\
\bm{v_1v_4}
& 0 & 0 & 0 & 0 & 0 & 0 & 0 & 0 & 0 & 1 & 0 & 1 & 0 & 1 & 0 & 1 \\
\bm{v_2v_3}
& 0 & 0 & 0 & 0 & 0 & 0 & 1 & 1 & 0 & 0 & 0 & 0 & 0 & 0 & 1 & 1 \\
\bm{v_1v_3}
& 0 & 0 & 0 & 0 & 0 & 1 & 0 & 1 & 0 & 0 & 0 & 0 & 0 & 1 & 0 & 1 \\
\bm{v_1v_2}
& 0 & 0 & 0 & 1 & 0 & 0 & 0 & 1 & 0 & 0 & 0 & 1 & 0 & 0 & 0 & 1
\end{array}
\right]
$}, 
\]

Clearly the row indexed by the Boolean function $\boldsymbol{v}_i \boldsymbol{v}_j$ is the element-wise product of the rows indexed by $\boldsymbol{v}_i$ and $\boldsymbol{v}_j$.
\end{example}

The message symbols $\{s_1,\dots,s_k\}$, $k=\sum_{i=0}^r{m\choose i}$, are encoded into $(s_1,\dots,s_k)\boldsymbol{G}_{\mathrm{RM}(r,m)}=(c_1,\dots,c_n)$, $n=2^m$, and then stored across $n$ servers. 
Following the notations in~\cite{ref14_isit_ly2025service}, we say that a message symbol is of order $\ell$, $0\leq \ell \leq r$, 
if the index of the row corresponding to the message symbol in the generator matrix is a Boolean function of degree $\ell$.
Let $\mathcal{J}_{\ell}$ be the set of all ${m\choose \ell}$ message symbols of order $\ell$. According to the way we present the generator matrix of $\mathrm{RM}(r,m)$, it holds that $\mathcal{J}_\ell=\left[\sum_{i=0}^{\ell-1}{m\choose i}+1,\sum_{i=0}^{\ell}{m\choose i}\right]$.

Now we list some useful facts about the recovery sets for symbols of a given order, as presented in~\cite{ref14_isit_ly2025service}.

\begin{lemma}\label{lm1:rec sets}(\cite[Theorem 3--5]{ref14_isit_ly2025service})
For each message symbol of order $\ell$, $0\leq \ell \leq r$, it holds that
\begin{itemize}
  \item There is a unique recovery set $S$ of size $2^\ell$ which contains Server 1 and does not contain Server $2^m$.
  \item If $\ell < r$, then any other recovery set has size at least $2^{r+1}-2^{\ell} >2^r$. If $\ell =r$, then all the other recovery sets also have size $2^r$.
  \item The number of recovery sets of size $2^{r+1}-2^{\ell}$ is given by the Gaussian binomial coefficient
\(
\left[ {m-\ell \atop r+1-\ell} \right]_{2},
\)
where each server $j\in[n]\setminus S$ appears exactly 
\(
\left[ {m-\ell-1 \atop r-\ell} \right]_{2}
\)
times in these recovery sets.
\end{itemize}
\end{lemma}

\section{Service Rate Characterization for high-order Reed-Muller Codes}

Towards characterizing the SRR of high-order RM codes,
the following results are established in~\cite[Theorem 6,7]{ref14_isit_ly2025service} and its extended version~\cite[Theorem 7--9]{ref16_arxiv_ly2025service}.

\begin{lemma}\label{lm2:bound}
For a message symbol $s_i$ of order $\ell$, its maximum achievable demand is $\lambda_i\leq 1+\frac{\left[ {m-\ell \atop r-\ell-1} \right]_{2}}{\left[ {m-\ell-1 \atop r-\ell} \right]_{2}}=1+\frac{2^{m}-2^{\ell}}{2^{r+1}-2^{\ell}}$. The total achievable demand for all symbols of a given order $\ell$ is also upper bounded by $\sum_{i\in\mathcal{J}_\ell}\lambda_i\leq 1+\frac{2^{m}-2^{\ell}}{2^{r+1}-2^{\ell}}$. 
The total achievable demand for all symbols, regardless of the orders, is upper bounded by $\sum_{i=1}^k \lambda_i\leq 1+\frac{2^m-1}{2^r}$.
\end{lemma}

Take $\mathrm{RM}(2,4)$ as an example. The authors of~\cite{ref14_isit_ly2025service} have proved the following inequalities: $\lambda_1\leq \frac{22}{7}$; $\lambda_i\leq \frac{10}{3}$ for $i\in[2,5]$; $\lambda_{i}\leq 4$ for $i\in[6,11]$; $\lambda_2+\lambda_3+\lambda_4+\lambda_5\leq \frac{10}{3}$; $\lambda_{6}+\dots+\lambda_{11}\leq 4$; and finally $\lambda_1+\lambda_{2}+\dots+\lambda_{11}\leq 4.75$. The first three types of inequalities are shown to be achievable and they constitute a \emph{maximal achievable simplex} $\mathcal{A}\subseteq \mathcal{S}(\boldsymbol{G})$, whereas the last inequality defines a \emph{minimal enclosing simplex} $\Omega$ where $\mathcal{S}(\boldsymbol{G})\subseteq \Omega$.  

In this paper, what we are attempting is to dig more into the SRR, by either constructing achievable regions beyond the \emph{maximal achievable simplex} $\mathcal{A}$ or providing new types of upper bound inequalities which are tighter than the \emph{minimal enclosing simplex} $\Omega$. We first focus on the service rate region of $\mathrm{RM}(r,r+1)$, where $m = r + 1$.

\subsection{The Complete Characterization of SRR of $\mathrm{RM}(r,r+1)$}

\begin{theorem}\label{thm:theorem1}
For $r \ge 2$, consider the $r$-th order binary Reed-Muller code $\mathrm{RM}(r,r+1)$ of dimension $k = 2^{r+1}-1$. Its service rate region is given by
\[
\begin{aligned}
 \mathcal{S}(\boldsymbol{G})\!
=\!
\left\{\!
\boldsymbol{\lambda} \in \mathbb{R}_{\ge 0}^{k}
:
\sum_{i=1}^{k} \lambda_i \le 2 \!
\right\}
\! 
=\!
\operatorname{conv}\bigl(\!\{ \boldsymbol{0}_k, \boldsymbol{v}_1, \dots, \boldsymbol{v}_k \}\!\bigr),   
\end{aligned}
\]
where $\boldsymbol{v}_i = 2\,\boldsymbol{e}_i$, for $i \in [k]$.
\end{theorem}

\begin{proof}
For any message symbol $s_j$, $j \in [k]$, of order $\ell$, let $R_{j,1}$ denote its unique recovery set of size $2^{\ell}$. 
By Lemma~\ref{lm1:rec sets}, it satisfies
$1 \in R_{j,1}$ and $2^{m} \notin R_{j,1}$.

Since $m = r+1$, Lemma~\ref{lm1:rec sets} further implies that symbol $s_j$ has exactly one recovery set of size $2^{m}-2^{\ell}$, which we denote by $R_{j,2}$.
The recovery sets $R_{j,1}$ and $R_{j,2}$ form a partition of $[2^{m}]$. Consequently, for the recovery set $R_{j,2}$ it holds that
\(
1 \notin R_{j,2},
\;
2^{m} \in R_{j,2}.
\)

Throughout the paper we frequently use the property of linear programming (LP) and its duality. Consider the LP instance with constraints in the form of Equations (\ref{equa:1a}) and (\ref{equa:1b}) and objective function $\sum_{i=1}^{k} \lambda_i$. Its dual LP is then 
\[
\begin{aligned}
&\min \quad \sum_{v=1}^{2^{m}} w_v,\\
&\text{s.t.} \quad
\sum_{v \in R_{j,q}} w_v \ge 1,
\quad \forall\, j \in [k],\ q \in \{1,2\},
\end{aligned}
\]
where server $v \in [2^m]$ is assigned a non-negative weight $w_v$.

Based on the structure of $R_{j,1}$ and $R_{j,2}$ for all \(j \in [k]\) derived above, a feasible solution $\boldsymbol{w}$ of the dual LP is
\[
w_v =
\begin{cases}
1, & v = 1 \text{ or } v = 2^{m},\\[4pt]
0, & \text{otherwise}.
\end{cases}
\]

Therefore, the optimal value of the original LP instance, which equals the optimal value of its dual LP according to the \emph{strong duality theorem} of linear programming, is then upper bounded by $\sum_{v=1}^{2^{m}} w_v = 2$. Hence, it holds that
\begin{small}
\begin{equation*}
\mathcal{S}(\boldsymbol{G})\subseteq\left\{\boldsymbol{\lambda} \in \mathbb{R}_{\ge 0}^{k}:\sum_{i=1}^{k} \lambda_i \le 2\right\}=\operatorname{conv}\bigl(\{ \boldsymbol{0}_k, 2\boldsymbol{e}_1, \dots, 2\boldsymbol{e}_k \}\bigr).
\end{equation*}
\end{small}

On the other hand, for each $j \in [k]$, since the recovery sets $R_{j,1}$ and $R_{j,2}$ form a partition of $[2^{m}]$,
the rate vector $\boldsymbol{\lambda}=2\boldsymbol{e}_j$ is achievable by assigning $\lambda_{j,1} = 1$ and $\lambda_{j,2} = 1$. Thus $2\boldsymbol{e}_j\in \mathcal{S}(\boldsymbol{G})$ for all $j\in[k]$. Trivially $\mathbf{0}_k\in\mathcal{S}(\boldsymbol{G})$ as well. By the
convexity of $\mathcal{S}(\boldsymbol{G})$, any convex combination
of these allocation vertices lies within $\mathcal{S}(\boldsymbol{G})$. Consequently,
\(
\operatorname{conv}\bigl(\{ \boldsymbol{0}_k, 2\boldsymbol{e}_1, \dots, 2\boldsymbol{e}_k \}\bigr)
\subseteq
\mathcal{S}(\boldsymbol{G}).
\)

Combining both directions, we conclude that
\begin{small}
\begin{equation*}
\mathcal{S}(\boldsymbol{G})=\left\{\boldsymbol{\lambda} \in \mathbb{R}_{\ge 0}^{k}:\sum_{i=1}^{k} \lambda_i \le 2\right\}=\operatorname{conv}\bigl(\{ \boldsymbol{0}_k, 2\boldsymbol{e}_1, \dots, 2\boldsymbol{e}_k \}\bigr)
\end{equation*}
\end{small}
for any $\mathrm{RM}(r,r+1)$. 
\end{proof}
 
\subsection{New Inequalities Involving only Two Symbols}
The total maximum achievable demand for any two message symbols of the same order is given
in~\cite{ref14_isit_ly2025service}. 
We extend the analysis to any two message symbols of different orders, by deriving new upper bounds and proving its achievability.
In this way we obtain the complete characterization of the SRR when confined to  two message symbols, which is the two-dimensional projections of the exact SRR polytope. 

Fix an RM code $\mathrm{RM}(r,m)$. Consider two message symbols $s_i$ and $s_j$, \(i,j \in[k]\), with orders $\ell_1$ and $\ell_2$, respectively, where $0\leq \ell_1 < \ell_2\leq r$. 
Let $T_i \subseteq\{\bm{v}_1,\bm{v}_2,\dots,\bm{v}_m\}$ with $|T_i|=\ell_1$ denote the \emph{index set} of the Boolean function corresponding to the message symbol $s_i$. Define $T_j$ similarly for $s_j$.

By Lemma~\ref{lm1:rec sets}, the message symbol $s_i$ has a unique recovery set of size $2^{\ell_1}$, denoted by $R_{i,1}$.
Similarly, $s_j$ has a unique recovery set $R_{j,1}$ of size $2^{\ell_2}$. 
Furthermore, we will need the following corollary of Lemma \ref{lm1:rec sets}, which is not explicitly proposed but can be inferred from the proofs in \cite{ref14_isit_ly2025service}.

\begin{corollary}\label{Cor1:T2}
Following the notations above, we have
\begin{itemize}
    \item If $T_i \subseteq T_j$, then $R_{i,1} \subseteq R_{j,1}$, $|R_{j,1} \setminus R_{i,1}| = 2^{\ell_2}-2^{\ell_1}$.
    \item Otherwise, let $|T_i \cap T_j| = k<\ell_1$. Then we have $|R_{i,1} \cap R_{j,1}| = 2^k \text{ and } |R_{j,1} \setminus R_{i,1}| = 2^{\ell_2}-2^k$.
\end{itemize}
\end{corollary}

\begin{example}
Consider $\mathrm{RM}(2,4)$. Message symbols $s_2$, $s_4$, and $s_6$ correspond to the row index sets 
$\{\boldsymbol{v}_4\}$, $\{\boldsymbol{v}_2\}$, and $\{\boldsymbol{v}_3,\boldsymbol{v}_4\}$, respectively.
Therefore, we have
\begin{itemize}
    \item \(
R_{2,1} = \{1,2\} \subseteq \{1,2,3,4\} = R_{6,1}, \;
\lvert R_{6,1} \setminus R_{2,1} \rvert = 2.
\)
\item \(
R_{4,1} = \{1,5\},\lvert R_{6,1} \cap R_{4,1} \rvert=1,\lvert R_{6,1} \setminus R_{4,1} \rvert = 3
\).
\end{itemize}
\end{example}

\begin{theorem}\label{thm:theorem2}
Consider two message symbols $s_i$ and $s_j$, with orders $\ell_1$ and $\ell_2$, respectively. 
Assume that $\ell_1 < \ell_2$. Then, the total achievable demand for $s_i$ and $s_j$ is upper bounded by
\[
\lambda_i + \lambda_j \le 1 + \frac{2^{m} - 2^{\ell_2}}{2^{r+1} - 2^{\ell_2}}.  \tag{2}\label{equa:2}
\]
\end{theorem}
\begin{proof}
By Lemma~\ref{lm1:rec sets} and the proofs in~\cite{ref14_isit_ly2025service}, for symbol $s_i$, any recovery set other than $R_{i,1}$ contains at least $2^{r+1}-2^{\ell_1}$ servers outside $R_{i,1}$, and hence has size at least $2^{r+1}-2^{\ell_1}$. The same holds for $s_j$.
Recall that $\{ R_{i,1}, \ldots, R_{i,t_i} \}$ denotes all $t_i$ recovery sets for the symbol $s_i$. Similarly, $\{ R_{j,1}, \ldots, R_{j,t_j} \}$ denotes all $t_j$ recovery sets for $s_j$.

Consider the LP with constraints in the form of (\ref{equa:1a}) and (\ref{equa:1b}) and objective function is $\lambda_i+\lambda_j$. Its dual LP is given by
\[
\begin{aligned}
&\min \quad \sum_{v=1}^{2^{m}} w_v,\\
&\text{s.t.} \quad
\sum_{v \in R_{d,q}}\! w_v \!\ge\! 1,\; \forall\, d \!\in\! \{i,j\},\ q\! \in\! [t_d].
\end{aligned}
\]

\textbf{Case 1:} $T_i \subseteq T_j$.  
It follows from Corollary~\ref{Cor1:T2} that $R_{i,1}\subseteq R_{j,1} \text{ and }|R_{j,1} \setminus R_{i,1}| = 2^{\ell_2}-2^{\ell_1}$, we can construct a feasible solution to the dual LP as
\[
w_v =
\begin{cases}
1, & v = 1,\\[4pt]
0, & v \in R_{j,1} \text{ and } v \neq 1,\\[6pt]
\dfrac{1}{2^{r+1}-2^{\ell_2}}, & v \notin R_{j,1}.
\end{cases}
\]

Now we explain why it is a feasible solution. 
For \(R_{i,1}\) and \(R_{j,1}\), the constraint is immediately satisfied since both recovery sets contain Server~1 with assigned weight \(w_1=1\).
For symbol \(s_j\), consider the remaining recovery sets \(R_{j,q}\), \(q\in [2,t_j]\),
each containing at least \(2^{r+1}-2^{\ell_2}\) servers outside \(R_{j,1}\). Such servers are assigned a  weight \(w_v=\frac{1}{2^{r+1}-2^{\ell_2}}\), ensuring the constraint $\sum_{v\in R_{j,q}} w_v \ge 1$, for $q\in[2,t_j]$.

For symbol \(s_i\), each recovery set \(R_{i,q}\) with \(q\in[2,t_i]\) contains at least
\(2^{r+1}-2^{\ell_1}\) servers outside \(R_{i,1}\).
Among these servers, there are at most \(|R_{j,1}\setminus R_{i,1}|=2^{\ell_2}-2^{\ell_1}\) servers with weight $0$. Hence,
\[
\sum_{v\in R_{i,q}} w_v \ge 
\frac{(2^{r+1}-2^{\ell_1})-(2^{\ell_2}-2^{\ell_1})}{2^{r+1}-2^{\ell_2}} = 1, \forall q\in[2,t_i]
\]
and thus all dual constraints are satisfied.

\textbf{Case 2:} $T_i\nsubseteq T_j$.
By Corollary~\ref{Cor1:T2},
let $k = |T_i\cap T_j|$,
then $R_{i,1} \cap R_{j,1} = M$, where $|M| = 2^{k}$.
Let $N \subseteq R_{j,1} \setminus M$ be any subset of size $2^{\ell_1}-2^{k}$.
In this case, a feasible solution to the dual LP is given by
\[
w_v =
\begin{cases}
1-\dfrac{2^{\ell_1}-2^k}{2^{r+1}-2^{\ell_2}}, & v = 1,\\[6pt]
\dfrac{1}{2^{r+1}-2^{\ell_2}}, & v \notin R_{j,1}\text{ or }v \in N,\\[4pt]
0, & \text{otherwise}.
\end{cases}
\]
Similar as the previous case, it can be verified that all constraints are satisfied by this solution. Due to space limitations the details are omitted.

Both feasible solutions lead to
\(
\sum w_v
= 1 + \frac{2^{m}-2^{\ell_2}}{2^{r+1}-2^{\ell_2}}
\). According to the strong duality theorem of linear programming, we have
the upper bound $\lambda_i + \lambda_j \leq\sum w_v
= 1 + \frac{2^{m}-2^{\ell_2}}{2^{r+1}-2^{\ell_2}}$. 
\end{proof}

\begin{figure*}[!t]
\centering
\begin{subfigure}[b]{0.21\textwidth}
    \centering
    \includegraphics[width=\textwidth]{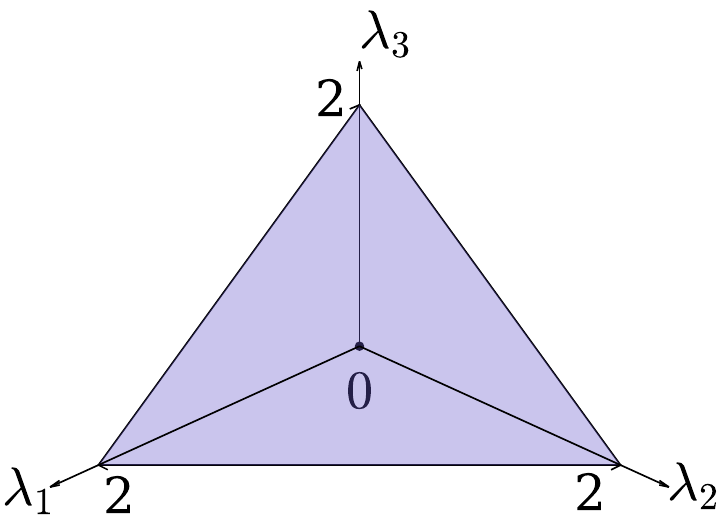}
    \caption{$\mathcal{S}(\boldsymbol{G})$ of $\mathrm{RM}(r,r+1).$}
    \label{fig(a):rm_(r,r+1)}
\end{subfigure}
\hfill
\begin{subfigure}[b]{0.28\textwidth}
    \centering
    \includegraphics[width=\textwidth]{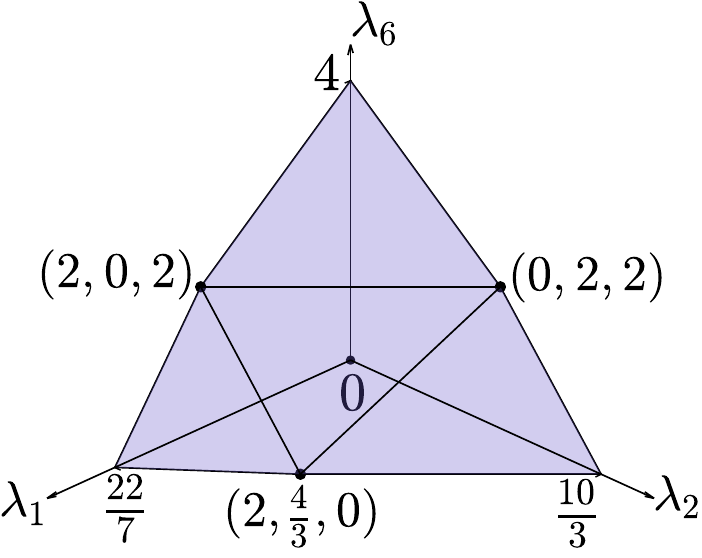}
    \caption{$\mathcal{B}$ of $\mathrm{RM}(2,4)$.}
    \label{fig(b):rm24}
\end{subfigure}
\hfill
\begin{subfigure}[b]{0.26\textwidth}
    \centering
    \includegraphics[width=\textwidth]{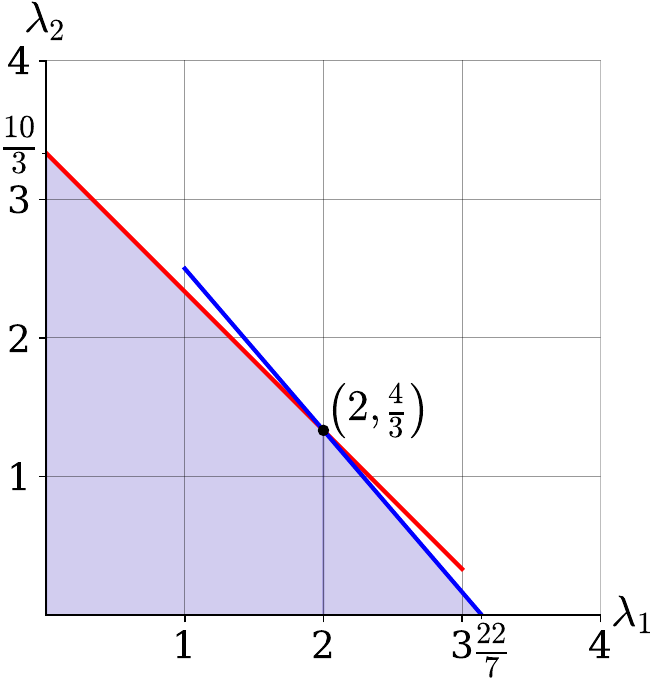}
    \caption{$\mathcal{S}_{1,2}(\boldsymbol{G})$ of $\mathrm{RM}(2,4).$}
    \label{fig(c):s1+2}
\end{subfigure}
\label{fig:three_planes}
\end{figure*}

Besides Constraint~\eqref{equa:2}, we next  use a capacity argument to derive an additional constraint on $\lambda_i$ and $\lambda_j$.
\begin{theorem}\label{thm:theorem3}


Any achievable SRR vector $\boldsymbol{\lambda} = (\lambda_1, \lambda_2, \dots, \lambda_k)$ must satisfy the following constraint: for any two message symbols $s_i$ and $s_j$ with orders $\ell_1 < \ell_2$, where $i,j \in [k]$, their achievable service rates $\lambda_i$ and $\lambda_j$ satisfy
\begin{equation*}
(2^{r+1}-2^{\ell_1})\lambda_i
+ (2^{r+1}-2^{\ell_2})\lambda_j \le 2^{m} + 2^{r+1} - 2^{\ell_1+1}.
\tag{3}\label{equa:3 Add Bound}
\end{equation*}

\end{theorem}

\begin{proof}
Recall that each of $s_i$ and $s_j$ has a unique recovery set of sizes $2^{\ell_1}$ and $2^{\ell_2}$, denoted by $R_{i,1}$ and $R_{j,1}$, respectively.
We assume that the recovery set $R_{i,1}$ is assigned a rate $a$, where
$0 \le a \le 1$.
Since $R_{i,1}$ and $R_{j,1}$ share the first server, the rate assigned to $R_{j,1}$ is at most
$1-a$.

For symbol $s_i$, all recovery sets other than $R_{i,1}$ have cardinality at least
$2^{r+1}-2^{\ell_1}$.
Therefore, the minimum system capacity used by the demand $\lambda_i$ is
\(
2^{\ell_1} a + \bigl(2^{r+1}-2^{\ell_1}\bigr)(\lambda_i - a).
\)
Similarly, for symbol $s_j$, the minimum capacity used by the demand $\lambda_j$ is
\(
2^{\ell_2}(1-a) + \bigl(2^{r+1}-2^{\ell_2}\bigr)\bigl(\lambda_j-(1-a)\bigr).
\)

Since there are $2^m$ servers, each server has uniform capacity $1$, the total
capacity available is $2^m$.
Hence, any $(\lambda_i,\lambda_j)$ in the service rate region must satisfy the following
constraint:
\begin{equation*}
\begin{aligned}
&2^{\ell_1} a + \bigl(2^{r+1}-2^{\ell_1}\bigr)(\lambda_i - a) \\
& + 2^{\ell_2}(1-a)
+ \bigl(2^{r+1}-2^{\ell_2}\bigr)\bigl(\lambda_j-(1-a)\bigr)
\le 2^m.
\end{aligned}
\end{equation*}

Rearranging the above inequality yields a characterization of the feasible region in
terms of $\lambda_i$ and $\lambda_j$:
\begin{equation*}
\begin{aligned}
&(2^{r+1}-2^{\ell_1})\lambda_i
+ (2^{r+1}-2^{\ell_2})\lambda_j \\
&\le 2^{m} + 2^{r+1} - 2^{\ell_2+1}
+ 2\bigl(2^{\ell_2}-2^{\ell_1}\bigr)a .
\end{aligned}
\end{equation*}

Here, $0 \le a \le 1$, and $r,m,\ell_1,\ell_2$ are constants.
Since $\ell_2 > \ell_1$, we have $2^{\ell_2}-2^{\ell_1} > 0$, it follows that the right-hand side of the inequality increases monotonically with $a$.
Consequently, the feasible region of $(\lambda_i,\lambda_j)$ is maximized when $a=1$, which yields the service rate constraint in Equation (\ref{equa:3 Add Bound}).
\end{proof}

Next we will demonstrate that the polytopes characterized by the two previous theorems are indeed achievable.

\begin{theorem}\label{thm:theorem4}
Consider the system $\mathrm{RM}(r,m)$, which provides service only to the symbols $s_i$ and $s_j$ of orders $\ell_1$ and $\ell_2$, respectively, with $\ell_1 < \ell_2$.
The associated two-symbol SRR, which is the projection of the full region $\mathcal{S}(\boldsymbol{G})$ onto the $(\lambda_i,\lambda_j)$-plane, is given by
\[
\mathcal{S}_{i,j}(\boldsymbol{G})
= \operatorname{conv}\!\left(
\left\{
\boldsymbol{0}_k,\;
\lambda_i^{\max}\boldsymbol{e}_i,\;
\lambda_j^{\max}\boldsymbol{e}_j,\;
\boldsymbol{u}_{i,j}
\right\}
\right),
\]
where
\(
\lambda_i^{\max} = 1 + \frac{2^{m}-2^{\ell_1}}{2^{r+1}-2^{\ell_1}}, 
\lambda_j^{\max} = 1 + \frac{2^{m}-2^{\ell_2}}{2^{r+1}-2^{\ell_2}} \text{ and }
\boldsymbol{u}_{i,j} = 2\cdot\boldsymbol{e}_i +
(\frac{2^{m}-2^{r+1}}{2^{r+1}-2^{\ell_2}})\cdot\boldsymbol{e}_j
\)
\end{theorem}

\begin{proof}

By combining~\eqref{equa:2} with~\eqref{equa:3 Add Bound} proved above, the two-symbol feasible region is upper bounded by
\(
\bigl(\lambda_i^{*}, \lambda_j^{*}\bigr)
=
\left(
2,\;
\frac{2^{m}-2^{r+1}}{2^{r+1}-2^{\ell_2}}
\right),
\) which yields the vertex $\boldsymbol{u}_{i,j}$.
The vertices $\lambda_j^{\max}\boldsymbol{e}_j$ and $\lambda_i^{\max}\boldsymbol{e}_i$ follow directly from Lemma~\ref{lm2:bound}.

\noindent \underline{Achievability:}
For the achievability part, we provide explicit constructions for the aforementioned upper bounds.

For symbol $s_j$, We let $\mathcal{R}^{\ast}_j$ denote the collection of its recovery sets whose size is exactly
$2^{r+1}-2^{\ell_2}$. By Lemma~\ref{lm1:rec sets}, 
\(
\lvert \mathcal{R}^{\ast}_j \rvert
=
\left[{m-\ell_2 \atop r+1-\ell_2}\right]_2,
\)
and each server $j \in [2^m]\setminus R_{j,1}$ is contained in exactly
\(
\left[{m-\ell_2-1 \atop r-\ell_2}\right]_2
\)
recovery sets in $\mathcal{R}^{\ast}_j$. 
Thus, assigning a request rate of
\(   
\frac{1}{\left[{m-\ell_2-1 \atop r-\ell_2}\right]_2}
\)
to each recovery set in $\mathcal{R}^{\ast}_j$ does not violate capacity constraint. In addition, 
since $R_{j,1}$ is disjoint from all sets in $\mathcal{R}^{\ast}_j$, we can assign a request rate of $1$ to $R_{j,1}$.
 We obtain
\(
\lambda_j^{\max}
=
1 + \frac{\lvert \mathcal{R}^{\ast}_j \rvert}{\left[{m-\ell_2-1 \atop r-\ell_2}\right]_2}
=
1 + \frac{2^{m}-2^{\ell_2}}{2^{r+1}-2^{\ell_2}}.
\)
Therefore, $\lambda_j^{\max}\boldsymbol{e}_j$ is achievable.
By the same argument, the demand vector $\lambda_i^{\max}\boldsymbol{e}_i$ is also achievable.

Next, we prove the achievability of vertex $\boldsymbol{u}_{i,j}$. For symbol $s_j$, we assign request
rate \(0\) to \textbf{any}
\(
\left[{m-\ell_2-1 \atop r-\ell_2}\right]_2
\)
recovery sets in $\mathcal{R}^{\ast}_j$ as well as to $R_{j,1}$, and assign a request rate of
\(
\frac{1}{\left[{m-\ell_2-1 \atop r-\ell_2}\right]_2}
\)
to each of the remaining recovery sets in $\mathcal{R}^{\ast}_j$. Thus
\(
\lambda_j^{*}
=\frac{\lvert \mathcal{R}^{\ast}_j \rvert-\left[{m-\ell_2-1 \atop r-\ell_2}\right]_2}{\left[{m-\ell_2-1 \atop r-\ell_2}\right]_2}
=
\frac{2^{m}-2^{r+1}}{2^{r+1}-2^{\ell_2}}
\) can be satisfied.

For symbol $s_i$, let $\mathcal{R}^{\ast}_i$ denote the collection of its recovery sets of size
$2^{r+1}-2^{\ell_1}$. Define $\mathcal{R}^{\ast}_{i,j} \subsetneq \mathcal{R}^{\ast}_i$ as the subset of recovery sets in $\mathcal{R}^{\ast}_i$ that contain $R_{j,1}\setminus\{1\}$.
The cardinality of $\mathcal{R}^{\ast}_{i,j}$ equals the number of $(r+1)$-dimensional subspaces containing $R_{j,1}$, which is
\(
\left[{m-\ell_2 \atop r+1-\ell_2}\right]_2,
\) which coincides with \(\lvert \mathcal{R}^{\ast}_j \rvert\).
Each recovery set $R_i^\prime \in \mathcal{R}^{\ast}_{i,j}$ corresponds to a unique $(r+1)$-dimensional subspace 
$\mathcal{V}^\prime$, and thus establishes a one-to-one correspondence with a recovery set $R_j^\prime \in \mathcal{R}^{\ast}_j$.
These sets satisfy
\(
R_i^\prime \cup R_{i,1} = R_j^\prime \cup R_{j,1} = \mathcal{V}^\prime.
\)
Accordingly, for $\left[{m-\ell_2-1 \atop r-\ell_2}\right]_2$ recovery sets in $\mathcal{R}^{\ast}_j$ that is assigned request rate \(0\), we assign a request rate of
\(
\frac{1}{\left[{m-\ell_2-1 \atop r-\ell_2}\right]_2}
\)
to the corresponding recovery sets in $\mathcal{R}^{\ast}_{i,j}$. Together with assigning a request rate of $1$ to
$R_{i,1}$,
\(
\lambda_i^{*}
=
1 + \frac{\left[{m-\ell_2-1 \atop r-\ell_2}\right]_2}
{\left[{m-\ell_2-1 \atop r-\ell_2}\right]_2}
=
2
\)
is also achievable. Therefore, the demand vector
\(
\boldsymbol{u}_{i,j}
=
2\,\boldsymbol{e}_i
+
\frac{2^{m}-2^{r+1}}{2^{r+1}-2^{\ell_2}}\,
\boldsymbol{e}_j
\)
can be satisfied.
By the convexity of SRR, the SRR associated with $s_i$ and $s_j$ is
\(
\mathcal{S}_{i,j}(\boldsymbol{G})
= \operatorname{conv}\Bigl(
\bigl\{
\boldsymbol{0}_k,\;
\lambda_i^{\max}\boldsymbol{e}_i,\;
\lambda_j^{\max}\boldsymbol{e}_j,\;
\boldsymbol{u}_{i,j}
\bigr\}
\Bigr).
\)
\end{proof}

\begin{remark}
For any $\mathrm{RM}(r,m)$ code, the maximal achievable convex polytope \(\mathcal{B} \subseteq \mathcal{S}(\boldsymbol{G})\) is given by
\[
\mathcal{B}\!=\!
\operatorname{conv}
\Bigl(\!
\{\boldsymbol{0}_k,\,
\boldsymbol{v}_i\!:\!\forall i \in [k],\,
\boldsymbol{u}_{i,j}\!:\!\forall i \in \mathcal{J}_{\ell_1}, j \in \mathcal{J}_{\ell_2}\}\!
\Bigr).
\]
where \(\boldsymbol{v}_i = \lambda_i^{max}\cdot \boldsymbol{e}_i\) and 
\(\boldsymbol{u}_{i,j}=2\cdot\boldsymbol{e}_i+(\frac{2^{m}-2^{r+1}}{2^{r+1}-2^{\ell_2}})\cdot\boldsymbol{e}_j\) with \(\forall \ell_1,\ell_2 \in [r],\, \ell_1 < \ell_2\). 
Consequently, $\mathcal{B}$ lies between $\mathcal{A}$ and $\Omega$ and provides a more accurate approximation of the SRR. 

For illustration, we consider $\mathrm{RM}(2,4)$, and the projection of $\mathcal{B}$ onto the $(\lambda_1,\lambda_2,\lambda_6)$ subspace is shown in Fig~(\subref{fig(b):rm24}).

For $\mathrm{RM}(r,r+1)$,
its SRR projected onto the $(\lambda_1,\lambda_2,\lambda_3)$ subspace is shown in Fig~(\subref{fig(a):rm_(r,r+1)}).
\end{remark}

\begin{example}
Consider the service rate region associated to $s_1$ and $s_2$ in the system $\mathrm{RM}(2,4)$. The symbols $s_1$ and $s_2$ correspond to message symbols of orders $0$ and $1$, respectively. By Lemma~\ref{lm2:bound}, we have
\(
\lambda_1^{\max} = \frac{22}{7}, \;
\lambda_2^{\max} = \frac{10}{3}.
\)

Substituting the parameters of $\mathrm{RM}(2,4)$ into inequalities~\eqref{equa:2} and~\eqref{equa:3 Add Bound} yields the following two constraints, which are represented by the red and blue line segments in Fig~(\subref{fig(c):s1+2}):
\(
\lambda_1 + \lambda_2 \le \frac{10}{3}, \;
7 \lambda_1 + 6 \lambda_2 \le 22.
\)
From these constraints, we obtain the vertex
\(
(\lambda_1, \lambda_2) = \left( 2, \frac{4}{3} \right).
\)

\noindent \underline{Achievability:}
Following Theorem~\ref{thm:theorem3}, The demand vector 
\(
(\lambda_1, \lambda_2) = \left(0, \frac{10}{3}\right)
\) 
and 
\(
(\lambda_1, \lambda_2) = \left(\frac{22}{7}, 0\right)
\) 
can be achieved easily. We construct a rate allocation that achieves
\(
(\lambda_1, \lambda_2)=\left(2,\frac{4}{3}\right)
\), the recovery sets of $s_1$ and $s_2$ are given as follows:
{\fontsize{8.5pt}{3pt}\selectfont
\[
\begin{array}{ll}
\!\boldsymbol{R_{1,1}\!=\!\{1\}}  
& R_{2,1}\!=\!\{1,2\} \\[2pt]
\!\boldsymbol{R_{1,2}\!=\!\{2,3,4,5,6,7,8\}} 
& R_{2,2}\!=\!\{3,4,5,6,7,8\} \\[2pt]
\!R_{1,3}\!=\!\{2,3,4,9,10,11,12\} 
& R_{2,3}\!=\!\{5,6,11,12,15,16\} \\[2pt]
\!R_{1,4}\!=\!\{2,3,4,13,14,15,16\} 
& R_{2,4}\!=\!\{7,8,11,12,13,14\} \\[2pt]
\!R_{1,5}\!=\!\{2,5,6,9,10,13,14\} 
& \boldsymbol{R_{2,5}\!=\!\{3,4,9,10,11,12\}} \\[2pt]
\!\boldsymbol{R_{1,6}\!=\!\{2,5,6,11,12,15,16\}} 
& \boldsymbol{R_{2,6}\!=\!\{3,4,13,14,15,16\}} \\[2pt]
\!R_{1,7}\!=\!\{2,7,8,9,10,15,16\} 
& \boldsymbol{R_{2,7}\!=\!\{5,6,9,10,13,14\}} \\[2pt]
\!\boldsymbol{R_{1,8}\!=\!\{2,7,8,11,12,13,14\}} 
& \boldsymbol{R_{2,8}\!=\!\{7,8,9,10,15,16\}} \\[2pt]
\cdots \\[2pt]
\!R_{1,16}\!=\!\{4,6,7,10,11,13,16\} \\[2pt]
\end{array}
\]
}
It can be verified that among all $15$ size-$7$ recovery sets of object $s_1$, 
there are exactly $7$ recovery sets of $s_1$ that contain 
$R_{2,1} \setminus \{1\}$, namely $R_{1,2}$ to $R_{1,8}$.
Moreover, these recovery sets are in one-to-one correspondence with the size-$6$ recovery sets of $s_2$.
For instance, $R_{1,6}$ corresponds uniquely to $R_{2,3}$, since
\(
R_{1,1} \cup R_{1,6}
=
R_{2,1} \cup R_{2,3}
=
\{1,2,5,6,11,12,15,16\}.
\)

Among the above-marked recovery sets, the rate allocation is
\(
\lambda_{1,1}=1,\;
\lambda_{1,2}=\lambda_{1,6}=\lambda_{1,8}
=\lambda_{2,5}=\lambda_{2,6}=\lambda_{2,7}=\lambda_{2,8}
=\tfrac{1}{3},
\)
with all remaining recovery sets assigned rate $0$. This yields
\(
\lambda_1=\sum_{j=1}^{16}\lambda_{1,j}=2,
\;
\lambda_2=\sum_{j=1}^{8}\lambda_{2,j}=\tfrac{4}{3}.
\)
Each server $i\in[16]$ utilizes $100\%$ of its capacity.

\begin{remark}
We present a new result of the total service rate of all symbols of $\mathrm{RM}(2,4)$, which is
\(
\sum_{j=1}^{11} \lambda_j \le \frac{13}{3}
\).

One feasible solution that achieves equality is given by
\(
\lambda_{1,1}=1,\;
\lambda_{6,2}=\lambda_{7,2}=\lambda_{8,2}=\lambda_{9,2}=\lambda_{10,3}=\lambda_{11,2}=\frac{1}{3},\;
\lambda_{7,3}=\lambda_{7,4}=\frac{2}{3}.
\)
The corresponding recovery sets are illustrated below. Servers \(2,3,5,9,16\) operate at a utilization level of \(66\%\), while all remaining servers are fully utilized.
{\fontsize{8.5pt}{3pt}\selectfont
\[
\begin{array}{ll}
\!\boldsymbol{R_{1,1}\!=\!\{1\}}  
\quad \boldsymbol{R_{6,2}\!=\!\{5,6,7,8\}} & \!\boldsymbol{R_{7,2}\!=\!\{3,4,7,8\}}\\[2pt]
\!\boldsymbol{R_{7,3}\!=\!\{9,10,13,14\}}  
& \!\boldsymbol{R_{7,4}\!=\!\{11,12,15,16\}} \\[2pt]
\!\boldsymbol{R_{8,2}\!=\!\{3,4,11,12\}}  
& \!\boldsymbol{R_{9,2}\!=\!\{2,4,6,8\}} \\[2pt] 
\!\boldsymbol{R_{10,3}\!=\!\{5,7,13,15\}}  
& \!\boldsymbol{R_{11,2}\!=\!\{2,6,10,14\}} \\[2pt] 
\end{array}
\]
}

Thus, we improves the earlier bound 
\(
\sum_{j=1}^{k} \lambda_j \le 1 + \frac{2^{m}-1}{2^{r}}
\) for \(\mathrm{RM}(2,4)\) in \cite{ref16_arxiv_ly2025service}, 
which is equal to $4.75$.   
\end{remark}

\end{example}

\bibliographystyle{IEEEtran}
\bibliography{ref}

\end{document}